\documentclass[a4paper]{llncs} 
\usepackage[utf8]{inputenc}
\usepackage{amssymb,amsmath}
\usepackage{todonotes}
\usepackage{bussproofs}
\usepackage{enumitem}

\renewenvironment{proof}{\begin{pf}}{\qed\end{pf}}

\newcommand{\prefix}{Q_1x_1\ldots Q_nx_n}

\newcommand{\var}{\mathsf{var}}
\newcommand{\poly}{\mathsf{poly}}

\newcommand{\ocef}{\mathrm{OBDD}(\land, \exists, \forall)}

\newtheorem{observation}{Observation}

\title{Proof Complexity of Symbolic QBF Reasoning}
\author{Stefan Mengel\inst{1}\thanks{Partially supported by the PING/ACK project of the French National Agency for Research (ANR-18-CE40-0011).} \and Friedrich Slivovsky\inst{2}\thanks{Supported by the Vienna Science and Technology Fund (WWTF) under grant ICT19-060.}}
\institute{
CNRS, UMR 8188, Centre de Recherche en Informatique de Lens (CRIL), Lens, F-62300, France\\
Univ. Artois, UMR 8188, Lens, F-62300, France
\and TU Wien, Vienna, Austria
	}

\begin{document}

\maketitle

\begin{abstract}
  We introduce and investigate symbolic proof systems for Quantified Boolean Formulas (QBF) operating on Ordered Binary Decision Diagrams (OBDDs).
  These systems capture QBF solvers that perform symbolic quantifier elimination, and as such admit short proofs of formulas of bounded path-width and quantifier complexity.
  As a consequence, we obtain exponential separations from standard clausal proof systems, specifically (long-distance) QU-Resolution and IR-Calc.

  We further develop a lower bound technique for symbolic QBF proof systems based on strategy extraction that lifts known lower bounds from communication complexity.
  This allows us to derive strong lower bounds against symbolic QBF proof systems that are independent of the variable ordering of the underlying OBDDs, and that hold even if the proof system is allowed access to an NP-oracle.
\end{abstract}

\section{Introduction}
Unlike in SAT solving, which is dominated by Conflict-Driven Clause Learning (CDCL), in QBF solving there is no single approach that is clearly dominant in practice. 
Instead, modern solvers are based on variety of techniques, such as (quantified) CDCL~\cite{ZhangM02,LonsingB10,PeitlSS19}, expansion of universal variables~\cite{Biere04,JanotaKMC16,BloemBHELS18}, and abstraction~\cite{RabeT15,JanotaM15,Tentrup16}.

In practice, these techniques turn out to be complementary, each having strengths and weaknesses on different classes of instances~\cite{PulinaT09,HoosPSS18,LonsingE18}.
This complementarity of solvers can be analyzed theoretically by considering proof complexity.
Essentially, the different paradigms used in solvers can be formalized as proof systems for QBF, which then can be analyzed with mathematical methods.
Then, by separating the strength of different proof systems, one can show that the corresponding solvers are unable to solve problems efficiently that can be dealt with by other solvers. This motivation has led to a great interest in QBF proof complexity over the last few years and resulted in a good understanding of common QBF proof systems and how they relate to each other (see~\cite{BeyersdorffCJ19,BeyersdorffBCP20} and the references therein).

In this paper, we focus on a \emph{symbolic} approach to QBF solving that was originally implemented in the {\sc QBDD} system~\cite{PanV04}.
Its underlying idea is to use OBDDs to represent constraints inside the solver, instead of clauses as used by most other SAT and QBF solvers.
% OBDDs have many good properties like canonicity and an efficient apply operation that then allow to perform variable elimination along a tree- or branch decomposition of the input formula.
We formalize {\sc QBDD} as a proof system in which the lines are OBDDs. More specifically, we consider QBF proof systems that are obtained from propositional OBDD-proof systems by adding $\forall$-reduction (cf.~\cite{BeyersdorffBCP20}).
%This approach is essentially that of the circuit-based Frege systems in, where we specialize the lines to be OBDDs.
%Due to the good properties, this leads to systems that are implementable in practice and in particular have already been implemented in~\cite{PanV04}.
Propositional proof systems using OBDDs as lines have been studied intensively since the introduction of this model in~\cite{AtseriasKV04}, see e.g.~\cite{BussIKS18}. We thus consider lifting these systems to QBF by adding $\forall$-reduction as very natural.

Analyzing the strength of OBDD-refutations, we first show that, even for a weak propositional system  that allows only conjunction of lines and forgetting of variables, the resulting QBF proof system, which we refer to as $\ocef$ and which corresponds to traces of {\sc QBDD}, $p$-simulates QU-resolution.
We also show that $\ocef$, and in fact also {\sc QBDD}, can make use of structural properties of QBF in the sense that instances of bounded pathwidth and bounded quantifier alternation can be solved efficiently. We do this by using a recent result on variable elimination for OBDDs from~\cite{CapelliM19} to show that the intermediate OBDDs in {\sc QBDD} are not too big in this setting.
We then observe that other QBF proof systems from the literature have hard instances of bounded pathwidth and bounded quantifier alternation. This shows that $\ocef$ can efficiently refute QBFs that are out of reach for many other systems. In particular, it is exponentially separated from (long-distance) QU-resolution~\cite{BalabanovWJ14} and the expansion based IR-calc~\cite{BeyersdorffCJ19}.
It follows that, at least in principle, {\sc QBDD} can solve instances that other, more modern solvers cannot.

The main technical contribution of this work is a lower bound technique for OBDD-refutations.
Here, we consider the strongest possible propositional system, which is semantic entailment of OBDDs.
We first show that this system admits efficient strategy extraction of decision lists whose terms are OBDDs.
Functions that can be succinctly encoded in this way have short protocols in a communication model from~\cite{ImpagliazzoW10} for which it is known that lower bounds can be obtained by proving that a function does not have large monochromatic rectangles.
To the best of our knowledge, such bounds are only known for fixed variable partitions.
To prove lower bounds for OBDD-refutations that are independent of the variable order chosen for the OBDDs, we lift classical bounds on the inner product function to a graph-based generalization which we show has essentially the same properties as the inner product function, but for \emph{all} variable partitions.

\section{Preliminaries}
\subsection{Propositional Logic and Quantified Boolean Formulas}
We assume an infinite set of propositional \emph{variables} and consider \emph{propositional formulas} built up from variables and the constants true ($1$) and false ($0$) using conjunction ($\land$), disjunction ($\lor$), and negation ($\neg$).
We write $\var(\varphi)$ for the set of variables occurring in a formula $\varphi$.
In particular, we are interested in formulas in \emph{conjunctive normal form (CNF)}. A formula is in CNF if it is a conjunction of \emph{clauses}. A clause is a disjunction of \emph{literals}, and a literal is variable $x$ or a negated variable $\neg x$.
An \emph{assignment} of a set $X$ of variables is a mapping $\tau: X \rightarrow \{0, 1\}$ of variables to truth values.
We write $[X]$ for the set of assignments of~$X$.
Given assignments $\tau: X \rightarrow \{0, 1\}$ and $\sigma: Y \rightarrow \{0, 1\}$ such that $X$ and $Y$ are disjoint, we let $\tau \cup \sigma$ denote the assignment of $X \cup Y$ such that $(\tau \cup \sigma)(x) = \tau(x)$ if $x \in X$ and $(\tau \cup \sigma)(x) = \sigma(x)$ if $x \in Y$.
Furthermore, we write $\tau|_{X'}$ for the restriction of $\tau$ to $X' \subseteq X$.
The result of applying an assignment~$\tau$ to formula~$\varphi$ and propagating constants is denoted $\varphi[\tau]$.
If $\varphi[\tau] = 1$ we say that $\tau$ \emph{satisfies} $\varphi$, and if $\varphi[\tau] = 0$, the assignment $\tau$ \emph{falsifies} $\varphi$.
A \emph{Quantified Boolean Formula (QBF)} is a pair $\Phi = \mathcal{Q}.\varphi$ consisting of a \emph{quantifier prefix} $\mathcal{Q}$ and a propositional formula $\varphi$, called the \emph{matrix} of $\Phi$. If the matrix is in CNF, then $\Phi$ is in \emph{Prenex Conjunctive Normal Form (PCNF)}.
The quantifier prefix is a sequence $\mathcal{Q} = Q_1x_1\ldots Q_nx_n$ where the $Q_i \in \{\forall, \exists\}$ are \emph{quantifiers} and the $x_i$ are propositional variables such that $\{x_1, \ldots, x_n\} = \var(\varphi)$.
We write $D_{\Phi}(x_i) = \{x_1, \dots, x_{i-1}\}$ for the set of variables that come before $x_i$ in the quantifier prefix, and say \emph{$x_i$ left of $x_j$} and \emph{$x_j$ is right of $x_i$} if $i < j$.
A variable $x_i$ is \emph{existential} if $Q_i = \exists$, and \emph{universal} if $Q_i = \forall$.
We write $\var_{\exists}(\Phi)$ for the set of existential variables, $\var_{\forall}(\Phi)$ for the set of universal variables, and $\var(\Phi)$ for the set of all variables occurring in $\Phi$.
Let $\Phi$ be a QBF. A \emph{universal strategy} for $\Phi$ is a family $\vec{f} = \{f_u\}_{u \in \var_{\forall}(\Phi)}$ of functions $f_u: [\var(\Phi)] \rightarrow \{0, 1\}$ such that $f_u(\tau) = f_u(\sigma)$ for any assignments $\tau$ and $\sigma$ that agree on $D_{\Phi}(u)$. If $\vec{f}$ is a universal strategy and $\tau: \var_{\exists}(\Phi) \rightarrow \{0, 1\}$ and assignment of existential variables, we write $\tau \cup \vec{f}(\tau)$ for the assignment of $\var(\Phi)$ such that $(\tau \cup \vec{f}(\tau))(x) = \tau(x)$ for existential variables $x \in \var_{\exists}(\Phi)$ and $(\tau \cup \vec{f}(\tau))(u) = f_u(\tau \cup \vec{f}(\tau))$ for universal variables $u \in \var_{\forall}(\Phi)$.
A universal strategy $\vec{f}$ is a \emph{universal winning strategy} for $\Phi$ if $\tau \cup \vec{f}(\tau)$ falsifies the matrix of $\Phi$ for every assignment $\tau$ of the existential variables.
A QBF is \emph{false} if it has a universal winning strategy, and \emph{true} otherwise.
\subsection{Graphs and Pathwidth of Formulas}
Let $G=(V,E)$ a graph and for every set $V'\subseteq V$ let $N[V']$ denote the open neighborhood of $V$, i.e., the set of all vertices in $V\setminus V'$ that have a neighbor in~$V'$. The \emph{expansion} of $G$ is then defined as $\min_{V'\subseteq V, |V'| \le |V|/2} \frac{|N(V')|}{|V'|}$.

A \emph{path decomposition} of a graph $G = (V, E)$ is a pair $(P, \lambda)$ where $P = p_1, \dots, p_n$ is a sequence of \emph{nodes} $p_i$, and $\lambda: \{p_1, \dots, p_n\} \rightarrow 2^V$ maps nodes $p_i$ to subsets $\lambda(p_i) \subseteq V$ of vertices called \emph{bags}, subject to the following constraints:
\begin{enumerate}
\item Each vertex appears in some bag, that is, $V \subseteq \bigcup_{i=1}^n \lambda(p_i)$,
\item For each edge $vw \in E$ there is a node $p_i$ such that $\{v,w\} \subseteq \lambda(p_i)$.
\item If $v \in \lambda(p_i)$ and $v \in \lambda(p_j)$ for $1 \leq i < j \leq n$, then $v \in \bigcap_{k=i}^j \lambda(p_k)$.
\end{enumerate}
The \emph{width} of a path decomposition is $\max_{i=1}^n |\lambda(p_i)| - 1$, and the \emph{pathwidth} of a graph $G$ is the minimum width of any path decomposition of $G$.
The pathwidth of a CNF formula $\varphi$ is the pathwidth of its \emph{primal graph}, which is the graph with vertex set $\var(\varphi)$ and edge set $\{xy\:|\: \exists C \in \varphi$ s.t. $x,y \in \var(C)\}$, and the pathwidth of a PCNF formula is the pathwidth of its matrix.

\subsection{OBDD}

We only give a short introduction into ordered binary decision diagrams (short OBDDs), a classical representation of Boolean functions~\cite{Bryant86}; see~\cite{Wegener00} for a textbook treatment.

Let $X$ be a set of variables and $\pi$ an ordering of~$X$. A $\pi$-OBDD on variables~$X$ is defined to be a directed acyclic graph $B$ with one source $s$ and two sinks labeled~$0$ and~$1$, called the $0$- and $1$-sink respectively. All non-sink nodes are labeled with variables from $X$ such that on every path $P$ in $B$ the variables appear in the order $\pi$. Moreover, all non-sink nodes have two outgoing edges, one labeled with $0$, the other with $1$. The size of $B$, denoted by $|B|$, is defined as the number of nodes in $B$.
% We denote by $\mathsf{OBDD_{<}}$ all OBDDs with order $<$.
Given an assignment $a\in \{0,1\}$, the OBDD $B$ computes a value $B(a)$ as follows: starting in the root, we construct a path by taking for every node $v$ labeled be a variable $x$ the edge labeled with $a(x)$. We continue until we end up in a sink, and the label of the sink is the value of $B$ on $a$ denoted by $B(a)$. This way $B$ computes a Boolean function and every Boolean function can be computed by an OBDD. The OBDD $B$ is called complete if on every source-sink path $P$ all variables in $X$ appear as node labels. The \emph{width} of a complete OBDD $B$ is defined as the maximal number of nodes that are labeled with the same variable.

\begin{observation}
 There is a polynomial time algorithm that, given an OBDD $B$, computes an equivalent complete OBDD $B'$. Moreover, $|B'| \le (|X|+1)|B|$.
\end{observation}

We will use the following algorithm of OBDDs which is often called the \emph{apply algorithm}.
\begin{lemma}\label{lem:OBDDproperties}
 Let $f:\{0,1\}^2 \rightarrow \{0,1\}$ be a binary Boolean function. Then there is an algorithm that, given two $\pi$-OBDDs $B_1$ and $B_2$, computes in time polynomial in $|B_1|+|B_2|$ a $\pi$-OBDD $B$ such that $B$ computes on input $a\in \{0,1\}^X$ the value $B(a):= f(B_1(a), B_2(a))$. In particular, the size of $B$ is polynomial in that of $B_1$ and $B_2$.
\end{lemma}

OBDDs are well-known to be \emph{canonical} in the sense that, for fixed variable order $\pi$, there is a unique representation of any Boolean function $f$ by a $\pi$-OBDD.
\begin{lemma}\label{lem:canonicity}
 Let $f$ be a Boolean function on variables $X$ and let $\pi$ be a variable order of $X$. Then there is a unique $\pi$-OBDD of minimal size (up to isomorphism) computing $f$. Moreover, given a $\pi$-OBDD representing $f$, this unique OBDD can be computed in polynomial time.
 The same is true for complete OBDDs.
\end{lemma}
%For every propositional formula $\varphi$, let $\piathsf{OBDD}(\varphi)$ denote the minimal $\pi$-OBDD representing $\varphi$.
Throughout this paper, we always assume that OBDDs are minimized with the help of the algorithm of Lemma~\ref{lem:canonicity}.
%\istefan{Ich habe erstmal was zusammengeschrieben, weiss aber nicht, ob das alles ist. Wir koennen ja einfach drauf achten beim Korrekturlesen.}

%We also use the fact that for every clase $C$, the OBDD-representation $\mathsf{OBDD_{<}}(\varphi)$ has size $|C|$ for every order $<$.

\subsection{Combinatorial Rectangles}
Let $X$ be a set of variables and $\Pi=(X_1, X_2)$ a partition of~$X$. We call $\Pi$ \emph{balanced} if $\min(|X_1|, |X_2|) \ge \lfloor|X|/2\rfloor$. More generally, we say that $\Pi$ is $b$-balanced if $\min(|X_1|, |X_2|) \ge b |X|$. A combinatorial rectangle with partition $\Pi$ is a function $R(X) = R_1(X_1)\land R_2(X_2)$. A dual way of seeing $R$ is defining $A$ to be the models of $R_1$ and $B$ those of $R_2$. Then the models of $R$ are exactly $A\times B$ and in a slight abuse of notation we then also write $R= A\times B$. A function $R$ is called a balanced rectangle if and only if $R$ is a combinatorial rectangle with a partition $\Pi$ that is balanced.

Let $f$ be a Boolean function and let $R$ be a combinatorial rectangle. We say that $R$ is monochromatic with respect to $f$ if either all models of $R$ are models of $f$ or no model of $R$ is a model of $f$. When $f$ is clear from the context, we simply call $R$ a monochromatic rectangle without remarking $f$ explicitly. We also say that $f$ \emph{has} the monochromatic rectangle $R$.

We will use the following well-known connection between OBDD and rectangles~\cite{KushilevitzN97}.
\begin{theorem} \label{thm:OBDDrectangles}
  Let $g$ be a function in variables $X$ computed by a $\pi$-OBDD of width~$w$. Let $X_1$ be a prefix of the variable order $\pi$ and let $X_2:= X\setminus X_1$.
  Then $g(X) = \bigvee_{i=1}^w R_i(X)$, where every $R_i$ is rectangle with partition $(X_1, X_2)$.
\end{theorem}

\section{Symbolic QBF Proof Systems}\label{sct:framework}
We consider line-based QBF proof systems where each line is an OBDD and is either the OBDD corresponding to a clause from the matrix, or else derived by a propositional proof system or universal reduction (cf. Frege systems with universal reduction~\cite{BeyersdorffBCP20}).
For simplicity, we will assume that each derivation begins with a sequence of OBDDs corresponding to the clauses in the matrix.

Formally, let~$\Phi = Q_1 x_1 \ldots Q_n x_n.C_1 \land \ldots \land C_m$ be a PCNF formula.
An \emph{OBDD derivation} of $L_k$ from $\Phi$ is a sequence $L_1, \dots, L_k$ of OBDDs, all with the same variable order $\pi$, such that each $L_i$ represents clause $C_i$ for $1 \leq i \leq m$, or is derived using one of the following rules:
\begin{enumerate}[itemsep=3pt,parsep=2pt]
\item \textbf{conjunction} $(\land)$: $L_i$ represents $L_j \land L_k$ for $j, k < i$.
\item \textbf{projection} $(\exists)$: $L_i$ represents $\exists x.L_j$ for some $x \in \var(L_j)$ and $j < i$.
\item \textbf{entailment} $(\models)$: $L_i$ is entailed by $L_{i_1}, \ldots, L_{i_k}$, for $i_1, \ldots i_k < i$.\footnote{Note that OBDD derivations using the entailment rule do not lead to proof systems in the sense of Cook and Reckhow~\cite{CookR79}, since checking entailment is $\mathsf{coNP}$-hard.}
\item \textbf{universal reduction} $(\forall)$: $L_i$ represents $L_j[u/c]$, where $j < i$, $u$ is a universally quantified variable that is rightmost among variables in $L_j$ and $c \in \{0, 1\}$.
\end{enumerate}
Here, $L_j[u/c]$ denotes the OBDD obtained from $L$ by removing each node labeled with variable $u$ and rerouting all incoming edges to its neighbor along the $c$-labeled edge (effectively substituting $c$ for $u$). The \emph{size} of an OBDD derivation is the sum of the sizes of the OBDDs in the derivation, and the \emph{width} of an OBDD derivation is the maximum width of any OBDD in the derivation.

It is not difficult to see that OBDD derivations are sound.
\begin{proposition}\label{prop:soundness}
  Let $L_1, \ldots, L_k$ be an OBDD derivation from $\Phi$. If $\Phi$ is true then $\prefix.L_1 \land \ldots \land L_k$ is true.
\end{proposition}

An \emph{OBDD-refutation} of $\Phi$ is an OBDD derivation of an OBDD representing~$0$.
A \emph{$\pi$-OBDD derivation} is an OBDD derivation where all OBDDs use variable order $\pi$.
We sometimes explicitly mention the derivation rules used in a proof. For instance, an $\ocef$ derivation is OBDD derivation using only conjunction, projection, and universal reduction.

\subsection*{A Proof System for Symbolic Quantifier Elimination}
We can use symbolic QBF proof systems to study the QBF solver {\sc QBDD} proposed by Pan and Vardi~\cite{PanV04}.
Given a PCNF formula~$\Phi = Q_1 x_1 \ldots Q_n x_n.\varphi$, {\sc QBDD} maintains \emph{buckets} $S_1, \dots, S_n$ of OBDDs such that $x_i$ is the rightmost variable (with respect to the quantifier prefix) occurring in the OBDDs of $S_i$.
Initially, the $S_i$ are the sets of clauses in $\varphi$ that have $x_i$ as their rightmost variable, represented as OBDDs.
{\sc QBDD} proceeds by eliminating variables from the inside out, starting with the variable~$x_n$.
To eliminate the variable~$x_i$, it computes the conjunction of OBDDs in bucket $S_i$, then removes $x_i$ from the result by quantifying either existentially or universally, depending on the quantifier $Q_i$.
The resulting OBDD is then added to the correct bucket.
The procedure terminates with a constant $1$ or constant $0$ OBDD, depending on whether the QBF~$\Phi$ is true or false.
Since any universal variable is innermost upon elimination, a run of {\sc QBDD} corresponds to an $\ocef$-derivation.

The aim of this section will be showing the following result:
\begin{proposition}\label{prop:twpanvardi}
 {\sc QBDD} solves PCNF formulas $\Phi$ with $q$ quantifier blocks and pathwidth $k$ in time $\mathsf{tower}(k, q+1)\;\poly(|\Phi|)$. 
\end{proposition}

Since, as stated above, the runs of {\sc QBDD} are proofs in $\ocef$, we directly get the following result on the strength of $\ocef$.

\begin{corollary}
 Every false PCNF $\Phi$ with $q$ quantifier width and pathwidth $k$ has an $\ocef$-refutation of size $\mathsf{tower}(k, q+1)\poly(|\Phi|)$.
\end{corollary}

As the basic tool, we use the following variable elimination result for OBDDs.

\begin{lemma}[\cite{CapelliM19}]
 Let $B$ be an OBDD of width $w$ and let $X$ be a subset of the variables in $B$. Then there is an OBDD~$B'$ of width $2^w$ that encodes $\exists X.B$ with the same variable order as~$B$. Moreover $B'$ can be computed in time $2^w \poly(|B|)$.
\end{lemma}

Note that since OBDD can be negated without increase of the representation size, we get that the same result is true for $\forall$-elimination. Iterating this result directly yields the following corollary.

\begin{corollary}\label{cor:tower}
 Let $B$ be an OBDD of width $w$ and let $Q_1 X_1 \ldots Q_q X_q$ a variable prefix with $q$ blocks. Then $Q_1 X_1 \ldots Q_q X_q\ F$ has an OBDD representation $B'$ of width $\mathrm{tower}(w, q)$. Moreover, $B'$ can be computed in time $\mathrm{tower}(w, q)\poly(|B|)$.
\end{corollary}
An analogous construction for the more general representation of structured DNNF~\cite{PipatsrisawatD08} is at the heart of the treewidth based QBF-algorithm in~\cite{CapelliM19}.

We can now proceed with the proof of Proposition~\ref{prop:twpanvardi}.
\begin{proof}[of Proposition~\ref{prop:twpanvardi}]
 Let $(P, \lambda)$ be a path decomposition of width $k$ of the primal graph of $F$. In~\cite{FerraraPV05} it is shown that there is an variable order $\pi$ depending only on $(P, \lambda)$ such that there is a complete OBDD of width $2^k$ computing $F$. Let $P_i := \bigwedge_{j\in [i]} S_i$. Then $P_i$ is the conjunction of some clauses of~$F$, so $(P,\lambda)$ is a path decomposition of $P_i$ of width at most $k$. It follows that for every $i\in [q]$ there is a complete OBDD representation of $P_i$ with order $\pi$ and width at most~$2^k$.
 
 We claim that all OBDD that are computed by {\sc QBDD} have width at most $\mathrm{tower}(k, q+1)$. Note first that all $S_i$ have pathwidth at most $k$ as above, so we can compute all of them by only conjoining OBDDs with order $\pi$ and of width at most $2^k$. Now whenever we eliminate a variable, the result is a a function that we get from $P_i$ by eliminating some variables. But since these variables are only in~$q$ quantifier blocks and we eliminate from the inside out, we have by Corollary~\ref{cor:tower} that the width of the result is at most $\mathrm{tower}(2^k,q) = \mathrm{tower}(k,q+1)$ which completes the proof. Noting that a complete OBDD of width $w$ in $n$ variables has size at most $wn$ and using canonicity and Lemma~\ref{lem:OBDDproperties} in all steps completes the proof.
\end{proof}

\section{Relation to Other Proof Systems}
In this section, we show that $\ocef$ is separated from several clausal QBF proof systems.
These results are obtained by identifying classes of QBFs that are hard for these proof systems but having bounded pathwidth and a fixed number of quantifier blocks.

We first consider \emph{Q-Resolution}~\cite{BuningKF95}, \emph{QU-Resolution}~\cite{Gelder12}, and \emph{Long-Distance Q-Resolution}~\cite{BalabanovJ12,EglyLW13}, which can be further generalized and combined into \emph{Long-Distance QU-Resolution}~\cite{BalabanovWJ14}.\footnote{This system is typically referred to as \emph{$LQU^+$-Resolution}.}
QU-Resolution allows resolution on universal pivots, Long-Distance Q-Resolution can derive tautological clauses in certain cases, and Long-Distance QU-Resolution additionally permits the derivation of tautological clauses by resolution on universal pivots (the proof rules of this system are shown in Figure~\ref{fig:lquplus} of Appendix~\ref{app:proofsystems}).

For all the proof systems above, we define the size of a refutation to be the number of clauses in it.
As usual, we say a proof system $\mathsf{P}$ $p$-simulates another proof system $\mathsf{P'}$ if for every proof $\Pi'$ in $\mathsf{P'}$ there is a proof $\Pi$ in $\mathsf{P}$ such that the length of $\Pi$ is polynomial in that of $\Pi'$.

\begin{proposition}
  $\ocef$ $p$-simulates QU-Resolution.
\end{proposition}
\begin{proof}
  We simulate QU-resolution line by line, using the fact that all clauses have small OBDD representations.
  An application of universal reduction in QU-resolution that removes literal~$l$ corresponds to an application of universal reduction in an OBDD derivation that replaces~$l$ by~$0$.
  Resolution of clauses $C_1 \lor x$ and $\neg x \lor C_2$ can be simulated by first computing an OBDD $L'$ representing $(C_1 \lor x) \land (\neg x \lor C_2)$. Each clause~$C$ can be represented by an OBDD of size $O(|C|)$, for any variable ordering, so by Lemma~\ref{lem:OBDDproperties}, the OBDD $L'$ can be computed in time polynomial in the size of the premises.
  To obtain an OBDD~$L$ representing the resolvent $C_1 \lor C_2$, we simply project out the pivot~$x$, that is, $L = \exists x.L'$.
\end{proof}

Lower bounds against QU-Resolution can be obtained by lifting lower bounds against bounded-depth circuits and decision lists~\cite{BeyersdorffCJ19,BeyersdorffBM20}.
This is because a decision list~\cite{Rivest87} encoding a universal winning strategy can be efficiently extracted from QU-Resolution refutations~\cite{BalabanovJ12}, and decision lists can be succinctly represented by bounded-depth circuits.
For instance, the class $\textsc{QParity}$ of formulas with the parity function as a unique universal winning strategy is hard for QU-Resolution~\cite{BeyersdorffCJ19}.
This class was modified so as to also demonstrate hardness for Long-Distance QU-Resolution, resulting in the class of formulas defined below.
\begin{align*}
  \textsc{QUParity}_n :=\: &\exists x_1 \ldots \exists x_n \forall z_1 \forall z_2 \exists t_2 \ldots \exists t_n. \\ &\textnormal{xor}_u(x_1,x_2,t_2,z_1,z_2) \land \textnormal{xor}_u(x_1, x_2, t_2, \neg z_1, \neg z_2)\: \land \\  & \bigwedge_{i=3}^n \left( \textnormal{xor}_u(t_{i-1},x_i,t_i,z_1,z_2) \land \textnormal{xor}_u(t_{i-1},x_i,t_i,\neg z_1, \neg z_2)\right) \: \land\\
                         &(z_1 \lor z_2 \lor t_n) \land (\neg z_1 \lor \neg z_2 \lor \neg t_n),
\end{align*}
where
\begin{align*}
  \textnormal{xor}_u(o_1,o_2,o,l_1,l_2) :=\: &(l_1 \lor l_2 \lor \neg o_1 \lor o_2 \lor o) \land (l_1 \lor l_2 \lor o_1 \lor \neg o_2 \lor o) \: \land \\
   &(l_1 \lor l_2 \lor \neg o_1 \lor \neg o_2 \lor \neg o) \land (l_1 \lor l_2 \lor o_1 \lor o_2 \lor \neg o).
\end{align*}
We restate the following result without a proof.
\begin{theorem}[\cite{BeyersdorffCJ19}]\label{thm:qparityhard}
$\textsc{QUParity}_n$ requires exponential-size refutations in Long-Distance QU-Resolution.
\end{theorem}
At the same time, the \textsc{QUParity} formulas have a very simple structure that can be exploited by symbolic proof systems.
\begin{lemma}
  The class $\{\textsc{QUParity}_n\}_{n \in \mathbb{N}}$ has bounded pathwidth.
\end{lemma}
\begin{proof}
  Let $n \in \mathbb{N}$ and consider the path $P = p_1, \ldots, p_n$ and node labeling $\lambda$ such that $\lambda(p_1) = \{x_1, x_2, t_2, z_1, z_2 \}$, $\lambda(p_i) = \{t_i, x_{i+1}, t_{i+1}, z_1, z_2\}$ for $2 \leq i < n$, as well as $\lambda(p_n) = \{z_1, z_2, t_n\}$.
  It is straightforward to verify that $(P, \lambda)$ is a path decomposition of $\textsc{QUParity}_n$, and its width is $4$.
% Moreover, the pathwidth of $\textsc{QParity}_n$ is at least~$2$ since it contains clauses on~$3$ variables.
\end{proof}
Since $\textsc{QUParity}_n$ only has three quantifier blocks, we obtain the following results by Proposition~\ref{prop:twpanvardi} and Theorem~\ref{thm:qparityhard}.
\begin{corollary}
  The formulas $\textsc{QParity}_n$ have polynomial-size $\ocef$ refutations.
\end{corollary}
\begin{theorem}
  QU-Resolution does not $p$-simulate $\ocef$.
\end{theorem}
Next, we look at the expansion-based proof system \emph{IR-calc}~\cite{BeyersdorffCJ19} (the proof rules are shown in Figure~\ref{fig:ircalc} of Appendix~\ref{app:proofsystems}).
For classes of formulas with a bounded number of quantifier blocks, lower bounds against IR-calc can be obtained by considering the \emph{strategy size}, which is the minimum range of any universal winning strategy (as a function mapping assignments of existential variables to assignments of universal variables)~\cite{BeyersdorffB20}.

\begin{definition}[Strategy Size~\cite{BeyersdorffB20}]
  The \emph{strategy size} $S(\Phi)$ of a false QBF~$\Phi$ is the minimum cardinality of the range of a universal winning strategy for $\Phi$.
%The strategy size of a PCNF family $\mathcal{F} = \{\Phi_i\}_{i \in \mathbb{N}}$ is the function $\triangledown_{\mathcal{F}}: \mathbb{N} \rightarrow \mathbb{N}$ that maps $n$ to the strategy size of $\Phi_n$.
%  $\vec{f}: \sigma \mapsto \vec{f}(\sigma)$
%  $\vec{f}: \{0, 1\}^{\var_{\exists}(\Phi)} \rightarrow \{0, 1\}^{\var_{\forall}(\Phi)}$ for $\Phi$.
\end{definition}
\begin{theorem}[\cite{BeyersdorffB20}]\label{thm:sizelowerbound}
  A false PCNF formula $\Phi$ with at most $k$ universal quantifier blocks requires IR-calc proofs of size $\sqrt[k]{S(\Phi)}$.
\end{theorem}
We use this correspondence to establish a proof size lower bound for the following class of formulas, which is a variant of the \emph{equality formulas}~\cite{BeyersdorffBH19} obtained by splitting the ``long'' clause $(t_1 \lor \ldots \lor t_n)$ into smaller clauses using auxiliary variables~$e_i$:
\begin{align*}
  \textsc{EQ}'_n := &\exists x_1 \ldots \exists x_n \forall u_1 \ldots \forall u_n \exists t_1 \ldots \exists t_n \exists e_1 \ldots \exists e_n.\\ & \bigwedge_{i=1}^n \left((x_i \lor u_i \lor \neg t_i) \land (\neg x_i \lor \neg u_i \lor \neg t_i)\right) \land \\
  &(t_1 \lor e_1) \land \bigwedge_{i=2}^{n-1} (\neg e_{i-1} \lor t_i \lor e_i) \land (\neg e_{n-1} \lor t_n)
\end{align*}
\begin{lemma}\label{lem:eqstrategy}
  $\textsc{EQ}'_n$ is false and the function $\vec{f}: \sigma \mapsto \vec{f}(\sigma)$ with $\vec{f}(\sigma)(u_i) = \sigma(x_i)$ for $1 \leq i \leq n $ is the unique universal winning strategy.
\end{lemma}
\begin{proof}
  Given any assignment $\sigma$ of the existential variables $x_i$, applying the joint assignment $\sigma \cup \vec{f}(\sigma)$ results in unit clauses $\bigwedge_{i=1}^n (\neg t_i)$, and unit propagation derives a contradiction. Thus $\vec{f}$ is a universal winning strategy and $\textsc{EQ}'_n$ is false.
Consider an assignment $\sigma$ of the $x_i$ together with an assignment $\tau$ of the $u_i$ such that $\sigma(x_i) \neq \tau(u_i)$ for some $i$.
  It is not difficult to see that the formula obtained by applying $\sigma \cup \tau$ can be satisfied by assigning the $t_i$ and $e_i$ appropriately, so the universal player can only win the evaluation game if they play according to~$\vec{f}$.
\end{proof}
\begin{proposition}
  Any IR-calc refutation of $\textsc{EQ}'_n$ has size $\Omega(2^n)$.
\end{proposition}
\begin{proof}
  By Lemma~\ref{lem:eqstrategy} the function $\vec{f}$ is the unique universal winning strategy for $\textsc{EQ}'_n$, and the cardinality of its range is $2^n$. Thus $2^n = S(\textsc{EQ}'_n)$ is a proof size lower bound for IR-calc by Theorem~\ref{thm:sizelowerbound}.
\end{proof}
\begin{lemma}
 The class $\{\textsc{EQ}'_n\}_{n \in \mathbb{N}}$ has bounded pathwidth.
\end{lemma}
\begin{proof}
  For $n \in \mathbb{N}$, we construct a path decomposition $(P, \lambda)$ of $\textsc{EQ}'_n$ as follows.
  We let $P = p_1, \ldots, p_n$ and define the labeling $\lambda$ as $\lambda(p_1) = \{x_1, u_1, t_1, e_1\}$, $\lambda(p_i) = \{e_{i-1}, x_i, u_i, t_i, e_i\}$ for $2 \leq i \leq n-1$, and $\lambda(p_n) = \{e_{n-1}, x_n, u_n, t_n\}$.
\end{proof}
\begin{corollary}
  The formulas $\textsc{EQ}'_n$ have polynomial-size $\ocef$ refutations.
\end{corollary}
\begin{theorem}
  IR-calc does not $p$-simulate $\ocef$.
\end{theorem}

\section{A Lower Bound on OBDD Refutations}
In this section, we present a technique for proving lower bounds on the size of OBDD-proofs even with the entailment $(\models)$ rule.
 We first show that such proofs admit efficient extraction of universal winning strategies as \emph{OBDD-decision lists}, a model which can in turn be efficiently transformed into \emph{rectangle decision lists}.
We then use a result by Impagliazzo and Williams~\cite{ImpagliazzoW10} to show that lower bounds for such decision lists reduce to size bounds of rectangles for Boolean functions. We complete the proof by deriving such a bound for a generalization of the well-known inner product function.
\subsection{From OBDD Proofs to Rectangle Decision Lists}
\begin{definition}
  Let $\mathcal{C}$ be a class of Boolean functions.
  A \emph{$\mathcal{C}$-decision list} of length~$s$ is a sequence $(L_1, c_1), \dots, (L_s, c_s)$ where the $c_i \in \{0, 1\}$ are truth values and the $L_i \in \mathcal{C}$ are circuits, and $L_s$ computes the constant function $1$. Let~$V$ be the set of variables occurring in the circuits~$L_i$.
The decision list computes a function $f: \{0, 1\}^V \rightarrow \{0, 1\}$ as follows. Given an assignment $\tau: V \rightarrow \{0, 1\}$, let $i = \min\{1 \leq j \leq s \;|\: L_j(\tau) = 1\}$. The we have $f(\tau) = c_i$.
\end{definition}
A \emph{$(w,\pi)$-OBDD-decision list} is a $\mathcal{C}$-decision list where $\mathcal{C}$ is the class of Boolean functions computed by $\pi$-OBDDs of maximum width~$w$.
Similarly, for a partition $(X,Y)$ of a set $V$ of variables, an \emph{$(X,Y)$-rectangle decision list} is a $\mathcal{C}$-decision list where $\mathcal{C}$ is the class of rectangles with respect to $(X,Y)$.

The next result states that OBDD-decision lists can be efficiently extracted from OBDD-proofs. Due to space constraints, its proof is in Appendix~\ref{app:strategyextraction}.
\begin{theorem}[Strategy Extraction~\cite{BalabanovJ12,BeyersdorffBCP20}]\label{thm:strategyextraction} There is a linear-time algorithm that takes a $\pi$-OBDD-refutation of a PCNF formula~$\Phi$ and outputs a family of $(w, \pi)$-OBDD-decision lists computing a universal winning strategy for~$\Phi$, where~$w$ is the width of the refutation.
\end{theorem}

\begin{lemma}\label{lem:rectanglelist}
  If there is a $(w, \pi)$-OBDD-decision list of size~$s$ computing a function $f: \{0,1\}^V \rightarrow \{0, 1\}$, and $(X,Y)$ is a bipartition of $V$ such that $X$ is the set of variables appearing in a prefix of $\pi$, then there is an $(X, Y)$-rectangle decision list of length $w (s-1) + 1$ computing $f$.
\end{lemma}
\begin{proof}
  Let $(L_1, c_1), \ldots, (L_s, c_s)$ be a $(w, \pi)$-OBDD-decision list computing function $f: \{0, 1\}^V \rightarrow \{0, 1\}$, and let $(X, Y)$ be a bipartition of $V$ such $X$ corresponds to the variables in a prefix of $\pi$. By Theorem~\ref{thm:OBDDrectangles}, each OBDD $L_i$ for $1 \leq i < s$ is equivalent to a disjunction $\bigvee_{j=1}^w R_{ij}(V)$ of rectangles with respect to $(X, Y)$.
  We construct an $(X,Y)$-rectangle decision list by replacing each pair $(L_i, c_i)$ for $1 \leq i < s$ by the sequence $(R_{i1}, c_i), \ldots, (R_{iw}, c_i)$. We can simply append $(L_s, c_s)$ to this sequence since the constant $L_s$ trivially is a rectangle.
  The resulting $(X,Y)$-rectangle decision list computes $f$ and has length $w(s-1)+1$.
\end{proof}
\subsection{From Rectangle Decision Lists to Communication Complexity}

We next use a result of Impagliazzo and Williams~\cite{ImpagliazzoW10} to prove lower bounds for rectangle decision lists.
% We note that the connection between the result in~\cite{ImpagliazzoW10} and rectangle decision lists has already been remarked by Chattopadhyay et al.~in~\cite{ChattopadhyayMM20}.
The following definition has been slightly simplified for our setting.
\begin{definition}
 Let $f$ be a Boolean function on variables $V$ and let $\Pi = (X, Y)$ be a partition of $V$. An \emph{AND-protocol} for $f$ with partition $\Pi$ is the following: two players are given an assignment to $X$ and $Y$, respectively, and want to compute $f$ on the joint assignment. To this end, they play in several rounds. In each round, they deterministically compute one bit each and send it to a third party. The third party computes the conjunction of the two bits and sends it to the players. If the conjunction evaluates to $1$, then the protocol ends and the players have to output the value of $f$ on the given input.
 
 The \emph{length} of the AND-protocol is the maximal number of rounds the players have to play to compute $f$ taken over all possible inputs for $f$.
\end{definition}

AND-protocols are interesting for us because of the following simple connection already observed without proof by Chattopadhyay et al.~\cite{ChattopadhyayMM20}.

\begin{proposition}\label{prop:rectanglesand}
 Let $f$ be a function in variables $V$ and let $\Pi$ be a partition of~$V$. If $f$ is computed by a rectangle decision list of length $s$ in which all rectangles have the partition~$\Pi$, then there is an AND-protocol for $\varphi$ with partition $\Pi$ of length at most $s$.
\end{proposition}
\begin{proof}
 The players simply evaluate the rectangle decision list: for every line $(R_{i}, c_i)$ where $R_i = R_{i,1}(X_1) \land R_{i,2}(X_2)$, the players evaluate $R_{i,1}$ and $R_{i,2}$ on their part of the input individually. Then the third party gives them the conjunction, so the value of the rectangle on the input. If it is $1$, then the players know that $f$ evaluates to $c_i$ on their input.
\end{proof}

Lower bounds on the length of AND-protocols can be shown thanks to the following result from~\cite{ImpagliazzoW10}.

\begin{theorem}\label{thm:andtosize}
 Let $f$ be a function in variables $V$ and let $\Pi$ be a balanced partition of $V$. If $f$ has an AND-protocol with partition $\Pi$ of length $s$, then there is a monochromatic rectangle with respect to $f$ with partition $\Pi$ of size at least $\frac{1}{4es} 2^{|V|}$.
\end{theorem}

\subsection{A Function with Only Small Monochromatic Rectangles}

With Theorem~\ref{thm:andtosize}, showing lower bounds for rectangle decision lists, and thus for OBDD-refutations, boils down to showing that functions to not have small monochromatic rectangles. Such function are known in the literature, see e.g.~\cite{KushilevitzN97}, but all results that we are aware of are for a fixed partition of the variables. However, since we want to show lower bounds independent of the choice of the variable order used in the OBDD-refutation, we need functions that have no big monochromatic rectangles for \emph{any} balanced partition of their variables.  We will construct such functions in this section.

The following result will be a building block in our construction.
\begin{proposition}\label{prop:matchingIP}
 Let $F:= \bigoplus_{i_\in [n]} g_i(x_i, y_i)$ where every function $g_i$ is either $g_i= x_i \land y_i$, $g_i = \neg x_i \land y_i$, $g_i = x_i \land \neg y_i$, or $g_i = x_i \lor y_i$. Then every monochromatic rectangle of $F$ has size at most $2^n$.
\end{proposition}

To show Proposition~\ref{prop:matchingIP}, we will use the following well known result from communication complexity: Let $\mathsf{IP}(x_1, \ldots, x_n, y_1, \ldots, y_n)$ be the inner product function defined as $\mathsf{IP}(x_1, \ldots, x_n, y_1, \ldots, y_n) := \bigoplus_{i\in [n]} x_i\cdot y_i$ where $\cdot$ denotes the multiplication over $\{0,1\}$ or equivalently conjunction. The following is well known, see e.g.~\cite{KushilevitzN97}.

\begin{lemma}\label{lem:ipclassical}
 All monochromatic rectangles of $\mathsf{IP}(x_1, \ldots, x_n, y_1, \ldots, y_n)$ have size at most $2^n$.
\end{lemma}

It is easy to see that the function $F$ from Proposition~\ref{prop:matchingIP} is a generalization of the inner product function. We will see that one can easily lift the bound on monochromatic rectangles.

\begin{proof}[of Proposition~\ref{prop:matchingIP}]
 First observe that $x_i \lor y_i = 1 \oplus (\neg x_i \land \neg y_i)$, so substituting every occurrence of $x_i\lor y_i$ by $\neg x_i\land \neg y_i$ will only change the color but not the size of any monochromatic rectangle. So in the remainder, we assume that there is no $g_i= x_i\lor y_i$ in $F$.
 
 In a next step, we substitute all occurrences of negated variables by the respective variables without the negation. Call the resulting formula $F'$. This substitution is clearly a bijection $\sigma$ between assignments that maintains the value, i.e., $F(X, Y) = F'(\sigma(X,Y))$. Since $\sigma$ acts on the variables independently, we have that for every monochromatic rectangle $A\times B$ of $F$, the set $\sigma(A\times B)$ is a monochromatic rectangle as well and $A\times B$ and $\sigma(A\times B)$ have the same size. Now observing that $F'$ is in fact the inner product function completes the proof using Lemma~\ref{lem:ipclassical}.
\end{proof}

We now introduce a generalization of $\mathsf{IP}$ with respect to an underlying graph structure. So let $X$ be a set of Boolean variables and let $G$ be a graph  with vertex set $X$ and edge set $E$. Then we define
\begin{align*}
 \mathsf{IP}_G(X) = \bigoplus_{xy\in E} x\cdot y.
\end{align*}
Note that with this definition $\mathsf{IP} = \mathsf{IP}_{M_n}$ where $M_n$ is a matching with $n$ edges.

\begin{lemma}\label{lem:rectanglesmall}
 Let $G=(X,E)$ be a graph with $n$ variables. Let $\{e_1, \ldots, e_m\}$ be an induced matching of $G$ and let $(X_1, X_2)$ be a partition of $X$ such that for every~$e_i$ one of the end points is in $X_1$ and one is in $X_2$. Then every monochromatic rectangle for $\mathsf{IP}_G$ respecting the partition $(X_1, X_2)$ has size at most $2^{n-m}$.
\end{lemma}
\begin{proof}
 Let $X'$ be the variables that are no end point in any of the $e_i$. Fix an assignment $a:X'\rightarrow \{0,1\}$. Let $e_i = x_iy_i$ and assume that $x_i\in X_1$ while $y_i\in X_2$. Let $\mathsf{IP}_{G,a}$ be the function in $X'' :=\{x_i, y_i\mid i\in [m]\}$ that we get from $\mathsf{IP}_G$ by plugging $a$ into the variables $X'$. Let $g_i$ be the function that, given an assignment $a_i$ to $x_i$ and $y_i$, counts the number of edges $e$ modulo $2$ that are incident to at least one of $x_i$ and $y_i$ and such that $a_i \cup a$ assigns $1$ to both end points of $e$. Clearly, $\mathsf{IP}_{G,a} = \bigoplus_{i\in [m]} g_i(x_i, y_i) \oplus c_a$ where $c_a\in \{0,1\}$ is a constant depending only on $a$. We will show that, up to the constant $c_a$ which does not change the size of monochromatic rectangles, the function $\mathsf{IP}_{G,a}$ has the form required by Proposition~\ref{prop:matchingIP}.
 
 To this end, let us analyze $g_i$. Let $X_i$ be the neighbors of $x_i$ different from $y_i$ and let $Y_i$ be the neighbors of $y_i$ different from $x_i$. Let $p_a(x_i)$ be the parity of variables in $X_i$ that are assigned $1$ by $a$ and let $p_a(y_i)$ be defined analogously for $y_i$. Then $g_i = (p_a(x_i)\land x_i) \oplus (p_a(y_i)\land y_i) \oplus (x_i\land y_i)$. We analyze the different cases:
 \begin{itemize}
 \item If $p_a(x_i)=0$ and $p_a(y_i) = 0$, then $g_i(x_i, y_i) = x_i \land y_i$.
 \item If $p_a(x_i) = 1$ and $p_a(y_i) = 0$, then $g_i(x_i, y_i) = x_i \oplus (x_i \land y_i)$. If $x_i=0$, then this term is $0$, so in all models we must have $x_i=1$. But $g_i(1, y_i)= 1\oplus y_i = \neg y_i$, so $g_i(x_i, y_i) = x_i \land \neg y_i$. 
 \item If $p_a(x_i)=0$ and $p_a(y_i)=1$, then $g_i(x_i, y_i)= \neg x_i \land y_i$ is obtained by a symmetric argument.
 \item Finally, if $p_a(x_i) = 1$ and $p_a(y_i) = 1$ then $g_i(x_i, y_i) = x_i \oplus y_i \oplus (x_i \land y_i)$. Clearly, if $x_i=y_i = 0$, then $g_i$ evaluates to $0$. Moreover, all other assignments evaluate to~$1$. So $g_i(x_i, y_i) = x_i \lor y_i$.
 \end{itemize}
 
 Thus, in any case, $g_i$ is of the form required by Proposition~\ref{prop:matchingIP}. It follows that every monochromatic rectangle of $\mathsf{IP}_{G,a}$ has size at most $2^m$.
 
 Now consider a monochromatic rectangle $R$ in $\mathsf{IP}$. Then, for every assignment $a:X'\rightarrow \{0,1\}$, restricting the variables $X'$ according to $a$ must give a monochromatic rectangle $R_a$ as well. It follows that 
 \begin{align*}
 |R| = \sum_{a:X'\rightarrow \{0,1\}} |R_a|.
 \end{align*}
But as we have seen, $|R_a|\le 2^m$. Moreover, there are $2^{|X'|} = 2^{n-2m}$ assignments to $X'$ and thus
\begin{align*}
 |R| \le 2^{n-2m} 2^m = 2^{n-m}
\end{align*}
as claimed.
 \end{proof}

\begin{theorem}\label{thm:expandersmallrectangles}
 Let $G=(X,E)$ be a graph with expansion $d$, degree $\Delta$ and $n$ vertices. Let $(X_1, X_2)$ be a $b$-balanced partition of $X$ for $b>d$. Then all monochromatic $(X_1, X_2)$-rectangles have size at most $2^{n \left(1-\frac{2nd^2b}{2(\Delta^2+1)}\right)}$.
\end{theorem}
\begin{proof}
 We show that there is an induced matching of size $\frac{nd^2b}{2(\Delta^2+1)}$ as in Lemma~\ref{lem:rectanglesmall}. Then the result follows directly.
 
 Assume w.l.o.g.~that $|X_1|\le |X_2|$. Then, by the expansion property of $G$, there are at least $d|X_1|$ neighbors of $|X_1|$ in $|X_2|$. Call these neighbors $X_2'$. Note that $X_2'$ has at least $d \min(\frac{|X|}2, |X_2'|)\ge d^2|X_1|$ neighbors in $X_1$ where the latter inequality is true because $d\le 1$. Denote the set of vertices in $X_1$ that have a neighbor in $X_2'$ by $X_1'$. Then $|X_1'| \ge d^2|X_1|$.
 
 We now construct a matching between $X_1'$ and $X_2'$. To this end, first delete all vertices not in $X_1'\cup X_2'$ from $G$. We then choose a matching iteratively as follows: pick a vertex $x_i\in X_1$ that has not been eliminated and that still has a neighbor $y_i$ in $X_2$. We add $x_iy_i$ to the matching and delete $x_i$ and $y_i$ and all their neighbors from $G$. If there are now any vertices in $X_i$ that have no neighbors outside of $X_i$ anymore, we delete those as well. We continue until $G$ is empty.
 
 We now analyze how many rounds we can make at least. First note that we delete at most $2\Delta$ neighbors of $x_i$ and $y_i$. Moreover, each of them can result in at most $\Delta-1$ vertices that have no neighbor on the other side of the partition anymore. So overall we delete at most $2\Delta^2+2$ vertices. Since we start with at least $2d^2|X_1|\ge 2d^2bn$ vertices, we can make $\frac{2d^2bn}{2(\Delta^2+1)}$ iterations before running out of vertices.
\end{proof}

\subsection{Putting It All Together}

In this section, we will finally show the promised lower bound for OBDD-refutations by putting together the results of the last sections.

\begin{theorem}
 There is an infinite sequence $(\Phi_n)$ of false PCNF formulas such that $|\Phi_n|= O(n)$ and every OBDD-refutation of $\Phi_n$ has size $2^{\Omega(n)}$.
\end{theorem}
\begin{proof}
  Choose a family of graphs of degree at most $\Delta$ and expansion $d$ for some constants $\Delta$ and $d$. Such families are well known to exist, see e.g.~\cite{hoory2006expander}. Out of this family, choose a sequence $(G_n)$ such that $G_n$ has $n$ vertices $X_n$. Now let $\varphi_n'= \neg \mathsf{IP}_{G_n}$. Clearly, $\varphi_n'$ can be computed by a Boolean circuit~$C_n$ of size $O(n)$. We apply Tseitin-transformation on that circuit to get a CNF formula~$\varphi_n$ that has as satisfying assignments exactly the values of all gates in $C_n$ under an assignment to inputs. Note that $\varphi_n$ has variables for all non-inputs of $C_n$ and thus in particular also for the output; let $z$ be the variable corresponding to the output of $C_n$ and let $Y$ denote the remaining variables of $\varphi_n$ introduced in the Tseitin-transformation. Then $\var(\varphi_n) = X_n \cup Y \cup \{z\}$. Moreover, $\varphi_n$ has size $O(n)$.
Now define \[\Phi_n = \exists X_n \forall z \exists Y \varphi_n.\]
 Then the only universal winning strategy $f_z$ is to return for every assignment $a$ to $X_n$ the negation of the value that $C_n$ evaluates to under $a$.  But then, using Theorem~\ref{thm:strategyextraction} and Lemma~\ref{lem:rectanglelist}, from every refutation of size $s$ and width $w$ of $\Phi_n$, we get a rectangle decision list of length $s'=w(s-1)+1$ for $\neg C_n = \mathsf{IP}_{G_n}$. Using Proposition~\ref{prop:rectanglesand} and Theorem~\ref{thm:andtosize}, we get that $\mathsf{IP}_{G_n}$ has a monochromatic rectangle of size $\frac{1}{4es'} 2^{|X_n|}= \frac{1}{4es'} 2^{n}$. But all monochromatic rectangles in $\mathsf{IP}_{G_n}$ have size at most $2^{n \left(1-\frac{2nd^2b}{2(\Delta^2+1)}\right)}$ by Theorem~\ref{thm:expandersmallrectangles}. Since $d, b$ and $\Delta$ are positive constants, it follows that $s' = 2^{\Omega(n)}$. But then at least one of $s$ and $w$ are in $2^{\Omega(n)}$, which gives the desired size bound.
\end{proof}
\section{Conclusion}

We have introduced OBDD-refutations that model symbolic OBDD-based reasoning for QBF. We have shown that these systems, already in the form that was used (implicitly) in a symbolic QBF solver~\cite{PanV04}, are surprisingly strong as they allow solving instances that are hard for the proof systems underlying state-of-the-art QBF solvers.
In view of this, it may be worthwhile to revisit these techniques in practice. There has been considerable progress in the computation of tree decompositions over the last few years (see e.g.~\cite{DellKTW17}) that could benefit a symbolic approach. Moreover, it could be interesting to use progress in knowledge compilation on generalizations of OBDDs that have similar properties but are in general exponentially more succinct. For example, one interesting candidate data structure might be \text{SDD}~\cite{Darwiche11}. While we consider it unlikely that such an approach would strictly beat current solvers, it might be sufficiently complementary to substantially improve the performance of a portfolio, much like the recently developed ADD-based symbolic model counter \textsc{ADDMC} has been shown to be highly complementary to DPLL-based state-of-the-art solvers~\cite{DudekPV20}. %So it would be interesting to see if symbolic techniques can also contribute to the state of the art for QBF solving.

We have also demonstrated limitations of OBDD-refutations by proving exponential lower bounds.
Our results require that all OBDDs appearing in a proof have the same variable order, but practical OBDD libraries such as \textsc{CUDD}~\cite{somenzi2009cudd} allow for dynamic variable reordering. While it is not clear how to use this to give more efficient refutations in an implementation of a QBF solver, it would be interesting to see if we can still show lower bounds in this generalized setting.
For refutations with variable reordering, the strategy extraction step and the transformation to rectangle decision lists go through unchanged, but there seems to be no equivalent of Theorem~\ref{thm:andtosize} for rectangle decision lists with varying partitions. It would be interesting to develop new techniques to show lower bounds in this setting.

\bibliographystyle{plain}
\bibliography{symbolic}

\begin{thebibliography}{10}

\bibitem{AtseriasKV04}
Albert Atserias, Phokion~G. Kolaitis, and Moshe~Y. Vardi.
\newblock Constraint propagation as a proof system.
\newblock In Mark Wallace, editor, {\em Principles and Practice of Constraint
  Programming - {CP} 2004, 10th International Conference, {CP} 2004, Toronto,
  Canada, September 27 - October 1, 2004, Proceedings}, volume 3258 of {\em
  Lecture Notes in Computer Science}, pages 77--91. Springer, 2004.

\bibitem{BalabanovJ12}
Valeriy Balabanov and Jie{-}Hong~R. Jiang.
\newblock Unified {QBF} certification and its applications.
\newblock {\em Formal Methods Syst. Des.}, 41(1):45--65, 2012.

\bibitem{BalabanovWJ14}
Valeriy Balabanov, Magdalena Widl, and Jie{-}Hong~R. Jiang.
\newblock {QBF} resolution systems and their proof complexities.
\newblock In Carsten Sinz and Uwe Egly, editors, {\em Theory and Applications
  of Satisfiability Testing - {SAT} 2014 - 17th International Conference, Held
  as Part of the Vienna Summer of Logic, {VSL} 2014, Vienna, Austria, July
  14-17, 2014. Proceedings}, volume 8561 of {\em Lecture Notes in Computer
  Science}, pages 154--169. Springer, 2014.

\bibitem{BeyersdorffB20}
Olaf Beyersdorff and Joshua Blinkhorn.
\newblock Lower bound techniques for {QBF} expansion.
\newblock {\em Theory Comput. Syst.}, 64(3):400--421, 2020.

\bibitem{BeyersdorffBH19}
Olaf Beyersdorff, Joshua Blinkhorn, and Luke Hinde.
\newblock Size, cost, and capacity: {A} semantic technique for hard random
  qbfs.
\newblock {\em Log. Methods Comput. Sci.}, 15(1), 2019.

\bibitem{BeyersdorffBM20}
Olaf Beyersdorff, Joshua Blinkhorn, and Meena Mahajan.
\newblock Hardness characterisations and size-width lower bounds for {QBF}
  resolution.
\newblock In Holger Hermanns, Lijun Zhang, Naoki Kobayashi, and Dale Miller,
  editors, {\em {LICS} '20: 35th Annual {ACM/IEEE} Symposium on Logic in
  Computer Science, Saarbr{\"{u}}cken, Germany, July 8-11, 2020}, pages
  209--223. {ACM}, 2020.

\bibitem{BeyersdorffBCP20}
Olaf Beyersdorff, Ilario Bonacina, Leroy Chew, and J{\'{a}}n Pich.
\newblock Frege systems for quantified boolean logic.
\newblock {\em J. {ACM}}, 67(2):9:1--9:36, 2020.

\bibitem{BeyersdorffCJ19}
Olaf Beyersdorff, Leroy Chew, and Mikol{\'{a}}s Janota.
\newblock New resolution-based {QBF} calculi and their proof complexity.
\newblock {\em {ACM} Trans. Comput. Theory}, 11(4):26:1--26:42, 2019.

\bibitem{Biere04}
Armin Biere.
\newblock Resolve and expand.
\newblock In {\em {SAT} 2004 - The Seventh International Conference on Theory
  and Applications of Satisfiability Testing, 10-13 May 2004, Vancouver, BC,
  Canada, Online Proceedings}, 2004.

\bibitem{BloemBHELS18}
Roderick Bloem, Nicolas Braud{-}Santoni, Vedad Hadzic, Uwe Egly, Florian
  Lonsing, and Martina Seidl.
\newblock Expansion-based {QBF} solving without recursion.
\newblock In Nikolaj Bj{\o}rner and Arie Gurfinkel, editors, {\em 2018 Formal
  Methods in Computer Aided Design, {FMCAD} 2018, Austin, TX, USA, October 30 -
  November 2, 2018}, pages 1--10. {IEEE}, 2018.

\bibitem{Bryant86}
Randal~E. Bryant.
\newblock Graph-based algorithms for boolean function manipulation.
\newblock {\em {IEEE} Trans. Computers}, 35(8):677--691, 1986.

\bibitem{BussIKS18}
Sam Buss, Dmitry Itsykson, Alexander Knop, and Dmitry Sokolov.
\newblock Reordering rule makes {OBDD} proof systems stronger.
\newblock In Rocco~A. Servedio, editor, {\em 33rd Computational Complexity
  Conference, {CCC} 2018, June 22-24, 2018, San Diego, CA, {USA}}, volume 102
  of {\em LIPIcs}, pages 16:1--16:24. Schloss Dagstuhl - Leibniz-Zentrum
  f{\"{u}}r Informatik, 2018.

\bibitem{CapelliM19}
Florent Capelli and Stefan Mengel.
\newblock Tractable {QBF} by knowledge compilation.
\newblock In Rolf Niedermeier and Christophe Paul, editors, {\em 36th
  International Symposium on Theoretical Aspects of Computer Science, {STACS}
  2019, March 13-16, 2019, Berlin, Germany}, volume 126 of {\em LIPIcs}, pages
  18:1--18:16. Schloss Dagstuhl - Leibniz-Zentrum f{\"{u}}r Informatik, 2019.

\bibitem{ChattopadhyayMM20}
Arkadev Chattopadhyay, Meena Mahajan, Nikhil~S. Mande, and Nitin Saurabh.
\newblock Lower bounds for linear decision lists.
\newblock {\em Chic. J. Theor. Comput. Sci.}, 2020, 2020.

\bibitem{CookR79}
Stephen~A. Cook and Robert~A. Reckhow.
\newblock The relative efficiency of propositional proof systems.
\newblock {\em J. Symb. Log.}, 44(1):36--50, 1979.

\bibitem{Darwiche11}
Adnan Darwiche.
\newblock {SDD:} {A} new canonical representation of propositional knowledge
  bases.
\newblock In Toby Walsh, editor, {\em {IJCAI} 2011, Proceedings of the 22nd
  International Joint Conference on Artificial Intelligence, Barcelona,
  Catalonia, Spain, July 16-22, 2011}, pages 819--826. {IJCAI/AAAI}, 2011.

\bibitem{DellKTW17}
Holger Dell, Christian Komusiewicz, Nimrod Talmon, and Mathias Weller.
\newblock The {PACE} 2017 parameterized algorithms and computational
  experiments challenge: The second iteration.
\newblock In Daniel Lokshtanov and Naomi Nishimura, editors, {\em 12th
  International Symposium on Parameterized and Exact Computation, {IPEC} 2017,
  September 6-8, 2017, Vienna, Austria}, volume~89 of {\em LIPIcs}, pages
  30:1--30:12. Schloss Dagstuhl - Leibniz-Zentrum f{\"{u}}r Informatik, 2017.

\bibitem{DudekPV20}
Jeffrey~M. Dudek, Vu~Phan, and Moshe~Y. Vardi.
\newblock {ADDMC:} weighted model counting with algebraic decision diagrams.
\newblock In {\em The Thirty-Fourth {AAAI} Conference on Artificial
  Intelligence, {AAAI} 2020, The Thirty-Second Innovative Applications of
  Artificial Intelligence Conference, {IAAI} 2020, The Tenth {AAAI} Symposium
  on Educational Advances in Artificial Intelligence, {EAAI} 2020, New York,
  NY, USA, February 7-12, 2020}, pages 1468--1476. {AAAI} Press, 2020.

\bibitem{EglyLW13}
Uwe Egly, Florian Lonsing, and Magdalena Widl.
\newblock Long-distance resolution: Proof generation and strategy extraction in
  search-based {QBF} solving.
\newblock In Kenneth~L. McMillan, Aart Middeldorp, and Andrei Voronkov,
  editors, {\em Logic for Programming, Artificial Intelligence, and Reasoning -
  19th International Conference, LPAR-19, Stellenbosch, South Africa, December
  14-19, 2013. Proceedings}, volume 8312 of {\em Lecture Notes in Computer
  Science}, pages 291--308. Springer, 2013.

\bibitem{FerraraPV05}
Andrea Ferrara, Guoqiang Pan, and Moshe~Y. Vardi.
\newblock Treewidth in verification: Local vs. global.
\newblock In Geoff Sutcliffe and Andrei Voronkov, editors, {\em Logic for
  Programming, Artificial Intelligence, and Reasoning, 12th International
  Conference, {LPAR} 2005, Montego Bay, Jamaica, December 2-6, 2005,
  Proceedings}, volume 3835 of {\em Lecture Notes in Computer Science}, pages
  489--503. Springer, 2005.

\bibitem{Gelder12}
Allen~Van Gelder.
\newblock Contributions to the theory of practical quantified boolean formula
  solving.
\newblock In Michela Milano, editor, {\em Principles and Practice of Constraint
  Programming - 18th International Conference, {CP} 2012, Qu{\'{e}}bec City,
  QC, Canada, October 8-12, 2012. Proceedings}, volume 7514 of {\em Lecture
  Notes in Computer Science}, pages 647--663. Springer, 2012.

\bibitem{hoory2006expander}
Shlomo Hoory, Nathan Linial, and Avi Wigderson.
\newblock Expander graphs and their applications.
\newblock {\em Bulletin of the American Mathematical Society}, 43(4):439--561,
  2006.

\bibitem{HoosPSS18}
Holger~H. Hoos, Tom{\'{a}}s Peitl, Friedrich Slivovsky, and Stefan Szeider.
\newblock Portfolio-based algorithm selection for circuit qbfs.
\newblock In John~N. Hooker, editor, {\em Principles and Practice of Constraint
  Programming - 24th International Conference, {CP} 2018, Lille, France, August
  27-31, 2018, Proceedings}, volume 11008 of {\em Lecture Notes in Computer
  Science}, pages 195--209. Springer, 2018.

\bibitem{ImpagliazzoW10}
Russell Impagliazzo and Ryan Williams.
\newblock Communication complexity with synchronized clocks.
\newblock In {\em Proceedings of the 25th Annual {IEEE} Conference on
  Computational Complexity, {CCC} 2010, Cambridge, Massachusetts, USA, June
  9-12, 2010}, pages 259--269. {IEEE} Computer Society, 2010.

\bibitem{JanotaKMC16}
Mikol{\'{a}}s Janota, William Klieber, Jo{\~{a}}o Marques{-}Silva, and
  Edmund~M. Clarke.
\newblock Solving {QBF} with counterexample guided refinement.
\newblock {\em Artif. Intell.}, 234:1--25, 2016.

\bibitem{JanotaM15}
Mikol{\'{a}}s Janota and Jo{\~{a}}o Marques{-}Silva.
\newblock Solving {QBF} by clause selection.
\newblock In Qiang Yang and Michael~J. Wooldridge, editors, {\em Proceedings of
  the Twenty-Fourth International Joint Conference on Artificial Intelligence,
  {IJCAI} 2015, Buenos Aires, Argentina, July 25-31, 2015}, pages 325--331.
  {AAAI} Press, 2015.

\bibitem{BuningKF95}
Hans Kleine{ }B{\"{u}}ning, Marek Karpinski, and Andreas Fl{\"{o}}gel.
\newblock Resolution for quantified boolean formulas.
\newblock {\em Inf. Comput.}, 117(1):12--18, 1995.

\bibitem{KushilevitzN97}
Eyal Kushilevitz and Noam Nisan.
\newblock {\em Communication complexity}.
\newblock Cambridge University Press, 1997.

\bibitem{LonsingB10}
Florian Lonsing and Armin Biere.
\newblock Depqbf: {A} dependency-aware {QBF} solver.
\newblock {\em J. Satisf. Boolean Model. Comput.}, 7(2-3):71--76, 2010.

\bibitem{LonsingE18}
Florian Lonsing and Uwe Egly.
\newblock Evaluating {QBF} solvers: Quantifier alternations matter.
\newblock In John~N. Hooker, editor, {\em Principles and Practice of Constraint
  Programming - 24th International Conference, {CP} 2018, Lille, France, August
  27-31, 2018, Proceedings}, volume 11008 of {\em Lecture Notes in Computer
  Science}, pages 276--294. Springer, 2018.

\bibitem{PanV04}
Guoqiang Pan and Moshe~Y. Vardi.
\newblock Symbolic decision procedures for {QBF}.
\newblock In Mark Wallace, editor, {\em Principles and Practice of Constraint
  Programming - {CP} 2004, 10th International Conference, {CP} 2004, Toronto,
  Canada, September 27 - October 1, 2004, Proceedings}, volume 3258 of {\em
  Lecture Notes in Computer Science}, pages 453--467. Springer, 2004.

\bibitem{PeitlSS19}
Tom{\'{a}}s Peitl, Friedrich Slivovsky, and Stefan Szeider.
\newblock Dependency learning for {QBF}.
\newblock {\em J. Artif. Intell. Res.}, 65:180--208, 2019.

\bibitem{PipatsrisawatD08}
Knot Pipatsrisawat and Adnan Darwiche.
\newblock New compilation languages based on structured decomposability.
\newblock In Dieter Fox and Carla~P. Gomes, editors, {\em Proceedings of the
  Twenty-Third {AAAI} Conference on Artificial Intelligence, {AAAI} 2008,
  Chicago, Illinois, USA, July 13-17, 2008}, pages 517--522. {AAAI} Press,
  2008.

\bibitem{PulinaT09}
Luca Pulina and Armando Tacchella.
\newblock A self-adaptive multi-engine solver for quantified boolean formulas.
\newblock {\em Constraints An Int. J.}, 14(1):80--116, 2009.

\bibitem{RabeT15}
Markus~N. Rabe and Leander Tentrup.
\newblock {CAQE:} {A} certifying {QBF} solver.
\newblock In Roope Kaivola and Thomas Wahl, editors, {\em Formal Methods in
  Computer-Aided Design, {FMCAD} 2015, Austin, Texas, USA, September 27-30,
  2015}, pages 136--143. {IEEE}, 2015.

\bibitem{Rivest87}
Ronald~L. Rivest.
\newblock Learning decision lists.
\newblock {\em Mach. Learn.}, 2(3):229--246, 1987.

\bibitem{somenzi2009cudd}
Fabio Somenzi.
\newblock {CUDD: CU decision diagram package-release 2.4. 0}.
\newblock {\em University of Colorado at Boulder}, 2009.

\bibitem{Tentrup16}
Leander Tentrup.
\newblock Non-prenex {QBF} solving using abstraction.
\newblock In Nadia Creignou and Daniel~Le Berre, editors, {\em Theory and
  Applications of Satisfiability Testing - {SAT} 2016 - 19th International
  Conference, Bordeaux, France, July 5-8, 2016, Proceedings}, volume 9710 of
  {\em Lecture Notes in Computer Science}, pages 393--401. Springer, 2016.

\bibitem{Wegener00}
Ingo Wegener.
\newblock {\em Branching Programs and Binary Decision Diagrams}.
\newblock {SIAM}, 2000.

\bibitem{ZhangM02}
Lintao Zhang and Sharad Malik.
\newblock Conflict driven learning in a quantified boolean satisfiability
  solver.
\newblock In Lawrence~T. Pileggi and Andreas Kuehlmann, editors, {\em
  Proceedings of the 2002 {IEEE/ACM} International Conference on Computer-aided
  Design, {ICCAD} 2002, San Jose, California, USA, November 10-14, 2002}, pages
  442--449. {ACM} / {IEEE} Computer Society, 2002.

\end{thebibliography}

\newpage
\appendix
\section{Soundness of Symbolic QBF Proof Systems}
\begin{proof}[of Proposition~\ref{prop:soundness}]
  We proceed by induction on the proof length $k$. For $1 \leq i \leq k$, let $\varphi_i = \bigwedge_{j=1}^i L_j$ denote the conjunction of proof lines up to $i$. If $k \leq m$ then the conjunction $\varphi_i \equiv \bigwedge_{i=1}^k C_i$ is logically equivalent to a subset of clauses of $\varphi$ and the result is immediate.
  For the induction step, if $L_k$ is derived by conjunction, projection, or entailment, then $L_1, \dots, L_{k-1} \models L_k$. Thus if $\prefix.\varphi_{k-1}$ is true, $\prefix. \varphi_k$ is true, and the result follows from the induction hypothesis. Otherwise, $L_k$ is derived from $L_i$ with $i < k$ by universal reduction, so that $L_k = L_i[u/c]$ for some universal variable $u$ and $c \in \{\bot,\top\}$.
  Towards a contradiction, assume that $\prefix.\varphi_{k-1}$ is true but $\prefix.\varphi_k$ is false. Let $\vec{f}$ be an existential winning strategy for $\prefix. \varphi_{k-1}$. Since $\prefix.\varphi_k$ is false, $\vec{f}$ is not a winning strategy for $\prefix.\varphi_k$, so there must be an assignment $\tau$ of $\{x_1, \dots, x_n\}$ that is consistent with $\vec{f}$ such that $\varphi_k[\tau] = 0$.
That is, $\tau$ falsifies a proof line $L_i$ with $1 \leq i \leq k$.
But $\vec{f}$ is a winning strategy of $\prefix.\varphi_{k-1}$ and thus $\varphi_{k-1}[\tau] = 1$, which leaves $L_k[\tau] = 0$ as the only option. Let $\tau'$ be an assignment that is consistent with $\vec{f}$ such that $\tau(x_j) = \tau'(x_j)$ for each variable $x_j <_{\Phi} u$ and $\tau'(u) = 1$ if $c = \top$ and $\tau'(u) = 0$ if $c = \bot$.
Such an assignment can be obtained from $\tau$ by setting the assignment of $u$ accordingly and ensuring the assignments of existential variables $e$ with $u <_{\Phi} e$ are consistent with $\vec{f}$.
Since $u$ is rightmost among variables in $L_i$ and $L_k = L_i[u/c]$ no longer contains $u$, we have $L_k[\tau'] = L_k[\tau] = 0$.
We further have $L_i[\tau'] = L_k[\tau']$ since $L_k = L_i[u/c]$ is obtained by substituting $c$ for $u$ and $c[\tau'] = \tau'(u)$.
That is, $L_i[\tau'] = 0$ and thus $\varphi_{k-1}[\tau'] = 0$ for an assignment $\tau'$ that is consistent with a winning strategy $\vec{f}$ of $\prefix.\varphi_{k-1}$, a contradiction.
\end{proof}

\section{Additional Clausal Proof Systems for QBF}\label{app:proofsystems}
\begin{figure}
  \fbox{\begin{minipage}[b]{\linewidth} \small
      	\begin{minipage}[b]{\linewidth}
				\medskip
				\begin{minipage}[b]{0.45\linewidth}
					\begin{prooftree}
						\AxiomC{}
						\RightLabel{(Axiom)}
						\UnaryInfC{$C$}
					\end{prooftree}
                                      \end{minipage}
                                      	\begin{minipage}[b]{0.45\linewidth}
					\begin{prooftree}
						\AxiomC{$D$}
						\RightLabel{(Universal Reduction)}
						\UnaryInfC{$D \setminus \{u, \neg u\}$}
					\end{prooftree}
                                      \end{minipage}

                                      Here, $C$ is a clause in the matrix and $u$ is a universal variable such that $D$ does not contain an existential variable that comes after $u$ (that is, ``depends on'' $u$) in the quantifier prefix.
                                      
			\end{minipage}
            \begin{minipage}[b]{\linewidth}
					\begin{prooftree}
                                          \AxiomC{$C_1 \lor U_1 \lor x$}
                                          \AxiomC{$C_2 \lor U_2 \lor \neg x$}
						\RightLabel{(Resolution)}
						\BinaryInfC{$C_1 \lor C_2 \lor U_1 \lor U_2$}
                                              \end{prooftree}
                                   \end{minipage}
                                   The pivot literal $x$ may be existential or universal. If $l \in C_1$ then $\overline{l} \notin C_2$ and vice versa.
Moreover, $U_1, U_2$ only contain universal literals with $\var(U_1) = \var(U_2)$ that come after the pivot~$x$ in the quantifier prefix.
\caption{The proof rules of Long-Distance QU-Resolution.\label{fig:lquplus}}
	\end{minipage}}
    \end{figure}
\begin{figure}
	\fbox{\begin{minipage}[b]{\linewidth} \small
            \begin{minipage}[b]{\linewidth}
              IR-calc operates on clauses containing \emph{annotated literals} that are pairs $(l, \sigma)$ where $l$ is a literal and $\sigma$ a partial assignment of universal variables. We write $l^\sigma$ for the annotated literal $(l, \sigma)$.
					\begin{prooftree}
						\AxiomC{}
						\RightLabel{(Axiom)}
						\UnaryInfC{$\{l^{[\tau]} \:|\: l \in C, l$ is an existential literal$\}$}
					\end{prooftree}
                                      \end{minipage}
                                     Here, $C$ is a clause in the matrix and $\tau$ the (minimal) partial assignment of universal variables that falsifies each universal literal in $C$. By $[\tau]$ we denote the restriction of $\tau$ to the universal variables that precede $\var(l)$ in the quantifier prefix.

			\begin{minipage}[b]{\linewidth}
				\medskip
				\begin{minipage}[b]{0.45\linewidth}
					\begin{prooftree}
						\AxiomC{$C_1 \vee e^\tau$}
						\AxiomC{$\neg e^\tau \vee C_2$}
						\RightLabel{(Resolution)}
						\BinaryInfC{$C_1 \vee C_2$}
					\end{prooftree}
                                      \end{minipage}
                                      	\begin{minipage}[b]{0.45\linewidth}
					\begin{prooftree}
						\AxiomC{C}
						\RightLabel{(Instantiation)}
						\UnaryInfC{$\mathsf{inst}(\tau, C)$}
					\end{prooftree}
                                      \end{minipage}
				\medskip
			\end{minipage}
                        The $C_i$ are clauses consisting of annotated literals. The pivot literals $e^\tau$ and $\neg e^\tau$ must have the same annotation $\tau$ in both premises. The instantiation allows us to extend literal annotations in the following way.
                        Given two assignments $\tau: X \rightarrow \{0, 1\}$ and $\sigma: Y \rightarrow \{0, 1\}$, let $\tau \circ \sigma: X \cup Y \rightarrow \{0, 1\}$ be the assignment such that $(\tau \circ \sigma)(x) = \tau(x)$ if $x \in X$ and $(\tau \circ \sigma)(x) = \sigma(x)$ if $x \in Y \setminus X$.
                        The clause $\mathsf{inst}(\tau, C)$ is defined as $\mathsf{inst}(\tau, C) = \{ l^{[\sigma \circ \tau]} \:|\: l^{\sigma} \in C\}$.
                        
\caption{The proof rules of IR-calc.\label{fig:ircalc}}
	\end{minipage}}
\end{figure}

\section{Strategy Extraction from OBDD Proofs}\label{app:strategyextraction}
\begin{proof}[of Theorem~\ref{thm:strategyextraction}]
  Let $R = L_1, \ldots, L_k$ be a $\pi$-OBDD-refutation of a PCNF formula~$\Phi$. For each universal variable~$u$, the algorithm is going to compute a $(w, \pi)$-OBDD decision list~$\mathsf{L}_u$ as follows. Let $L_{i_1} = L_{j_1}[u/c_1], L_{i_2} = L_{j_2}[u/c_2], \ldots$,  $L_{i_{\ell}} = L_{j_{\ell}}[u/c_{\ell}]$ be the lines of $R$ obtained by universal reduction of variable $u$ in their order of appearance in $R$, that is, $i_1 < i_2 < \ldots < i_\ell$ and $1 \leq j_r < i_r \leq k$ for each $r \in [\ell]$.
  The decision list is $\mathsf{L}_u = (\neg L_{i_1}, c_1), (\neg L_{i_2}, c_2), \ldots, (\neg L_{i_{\ell}}, c_{\ell}), (1, 1)$.
  These lists can be constructed in linear time by scanning the proof line by line and adding the pair $(\neg L_i, c)$ to the decision list $\mathsf{L}_u$ whenever $L_i = L_j[u/c]$ is derived from $L_j$ by universal reduction (recall that OBDDs can be negated simply by swapping the $0$ and $1$ sinks).
  It remains to show that the Boolean functions $\vec{f} = \{f_u\}_{u \in \var_{\forall}(\Phi)}$ computed by the decision lists $\mathsf{L}_u$ represent a winning universal strategy for $\Phi$.
  We begin by observing that, for every assignment $\tau$ of the existential variables, there is a unique assignment $\vec{f}(\tau)$ of the universal variables such that $\tau \cup \vec{f}(\tau)$ is consistent with $\vec{f}$: the OBDDs in each decision list $\mathsf{L}_u$ only contain variables that precede $u$ in the quantifier prefix, so that each function $f_u$ only depends on these variables and no circular dependencies can arise. The assignment $\vec{f}(\tau)$ can be computed simply by following the order of universal variables in the quantifier prefix.

  Let $m$ denote the number of clauses of $\Phi$. We now prove, by downward induction on $i$ for $m \leq i \leq k$, that $\varphi_i$ is falsified by any assignment $\tau$ that is consistent with $\vec{f}$, where $\varphi_i = \bigwedge_{j=1}^i L_j$ again denotes the conjunction of proof lines up to $i$. Since $\varphi_m$ is logically equivalent to the matrix of $\Phi$, this implies that $\vec{f}$ is a universal winning strategy.
  The base case $i=k$ is trivial as $L_k = \bot$ is falsified under any assignment $\tau \cup \vec{f}(\tau)$.
For the induction step, assume that the assignment $\tau \cup \vec{f}(\tau)$ falsifies $\varphi_i$ for each assignment $\tau$ of the universal variables.
We consider two cases:
\begin{enumerate}
\item If $L_i$ is derived using conjunction, projection, or entailment, then $\varphi_{i-1} \models L_i$, so any assignment that falsifies $L_i$ must falsify $\varphi_{i-1}$ as well.
  In combination with the induction hypothesis, this tells us that the assignment $\tau \cup \vec{f}(\tau)$ falsifies $\varphi_{i-1}$ for each assignment $\tau$ of the universal variables.

\item Otherwise, the line $L_i = L_j[u/c]$ is derived from $L_j$ with $j < i$ by universal reduction.
  Towards a contradiction, assume that there is an assignment $\tau$ of the universal variables so that $\tau \cup \vec{f}(\tau)$ satisfies $\varphi_{i-1}$ but falsifies $L_i$.
  Consider the decision list $\mathsf{L}_u$. By construction, it contains the pair $(\neg L_i, c)$, and since $\varphi_{i-1}$ is satisfied by $\tau \cup \vec{f}(\tau)$, the OBDD $\neg L_{i_r}$ is falsified for each pair $(\neg L_{i_r}, c_{i_r})$ that precedes $(\neg L_i, c)$ in $\mathsf{L}_u$. Since $L_i$ is falsified by $\tau \cup \vec{f}(\tau)$, the OBDD $\neg L_i$ is satisfied and thus $f_u(\tau \cup \vec{f}(\tau)) = c$.
  But $L_i = L_j[u/c]$ is falsified, so $L_j$ must be falsified as well, a contradiction.
  \end{enumerate}
\end{proof}

\end{document}